\documentclass[11pt]{article}

\usepackage{amssymb}

\newtheorem{thm}{Theorem}

\newtheorem{proposition}[thm]{Proposition}

\newenvironment{proof}{\noindent\bf{Proof.}\rm}{\hfill$\blacksquare$\bigskip}

\begin{document}

\title{NP-hardness of hypercube 2-segmentation}

\author{Uriel Feige~\thanks{Department of Computer Science, The Weizmann Institute, Rehovot, Israel. Email:
uriel.feige@weizmann.ac.il.}}

\maketitle

\begin{abstract}
The {\em hypercube 2-segmentation} problem is a certain biclustering problem that was previously claimed to be NP-hard, but for which there does not appear to be a publicly available proof of NP-hardness. This manuscript provides such a proof.
\end{abstract}

\section{Introduction}

We consider the following problem.

{\bf Hypercube 2-segmentation (H2S)}. Given a set of $k$ vectors $x_1, \ldots, x_k$ in $\{0,1\}^d$, one needs to select two centers $c_1, c_2$  in $\{0,1\}^d$ maximizing $$\sum_{i=1}^k \max[agree(c_1,x_i),agree(c_2,x_i)]$$
where $agree(x,y)$ counts on how many coordinates vectors $x$ and $y$ agree (which is $d$ minus the Hamming distance between $x$ and $y$).

H2S may also be phrased in the following equivalent way.

{\bf H2S -- $\ell_1$ maximization formulation.} Given a set of $k$ vectors $x_1, \ldots, x_k$ in $\{1,-1\}^d$, partition the $k$ vectors into two sets, such that the sum of the $\ell_1$ norms of the two corresponding vector sums is maximized.
\medskip

The equivalence between the two formulations of H2S follows from the fact that for a set of vectors on the hypercube, the location of the center that maximizes agreement is determined by taking the majority value on each coordinate separately. The $\ell_1$ norm of the sum shows how much this optimal center gains compared to placing a center at $0^d$ (the center of the $\{1,-1\}^d$ hypercube).

The following theorem was claimed in~\cite{KPR98} without proof.

\begin{thm}
\label{thm:KPR}
The hypercube 2-segmentation problem is NP-hard.
\end{thm}

In this manuscript we provide a proof of this theorem.

\subsection{Related work}

Kleinberg, Papadimitriou and Raghavan~\cite{KPR98} undertook a systematic study of the complexity of segmentation problems. Item~(4) of Theorem~2.1 in~\cite{KPR98} claims that H2S is NP-hard. The proof sketch of that theorem states that the proof is by reduction ``from {\em maximum satisfiability} with clauses that are equations modulo two", and gives no further details.
That paper also shows that given any set of vectors, at least one of these vectors is a nearly optimal center for the set: its agreement score with the set of vectors is at least a~$2\sqrt{2} - 2 \simeq 0.828$ fraction of that of the optimal center. This easily implies a polynomial time algorithm for approximating H2C within a ratio~$0.828$.

Alon and Sudakov~\cite{AS99} provide a PTAS for H2S. Specifically, for every choice of $\epsilon > 0$ they provide a linear time algorithm (with leading constant that depends on $\epsilon$) that approximates H2S with a factor no worse than $1 - \epsilon$. Similar results apply to hypergraph $k$-segmentation for constant $k > 2$.

In~\cite{KPR04}, the journal version of~\cite{KPR98}, Item (3) of Theorem~1.1 claims that H2S is MAXSNP-complete, and cites~\cite{KPR98} as reference, without providing a proof. This claim of MAXSNP-completeness contradicts the fact (proved in~\cite{AS99}) that H2S has a PTAS, and hence is incorrect. (The authors of~\cite{KPR04} do cite~\cite{AS99}.)

Wulff, Urner and Ben-David~\cite{WUB13} study a problem that they refer to as {\em monochromatic biclustering} (MCBC). They present a PTAS for (the maximization version of) MCBC, and
also prove NP-hardness of MCBC in the case that the input instance may contain {\em don't care} symbols. This NP-hardness result is based on a reduction from {\em max-cut}. In~\cite{WUB13}, it is conjectured that NP-hardness holds even without {\em don't care} symbols.
H2C is a special case of MCBC without {\em don't care} symbols, and hence Theorem~\ref{thm:KPR} implies the conjecture of~\cite{WUB13}. (Apparently, the authors of~\cite{WUB13} were unaware of the previous work on H2C cited above. The term {\em biclustering} does not appear in~\cite{KPR04}, whereas the term {\em segmentation} does not appear in~\cite{WUB13}.)

\section{Proof of NP-hardness}

We start with some background. A Hadamard code $H_M$ of dimension $M$ is a collection of $M$ vectors in $\{1, -1\}^M$ with the property that every two vectors are orthogonal. There are well known recursive constructions of Hadamard codes when $M$ is a power of~2, and hence we shall assume $M$ to be a power of~2.

Recall the notions of $\ell_1$ and $\ell_2$ norms of a vector. We shall use the following proposition.

\begin{proposition}
\label{pro:hadamard}
Consider an arbitrary set of distinct vectors from an arbitrary Hadamard code $H_M$. Then the $\ell_1$ norm of their sum is at most $M^{3/2}$.
\end{proposition}

\begin{proof}
The $\ell_2$ norm of a code word is $\sqrt{M}$. As the codewords are orthogonal, the $\ell_2$ norm of the sum of $q$ distinct vectors is $\sqrt{qM}$. The $\ell_1$ norm can exceed the $\ell_2$ norm by a factor of at most $\sqrt{M}$. As $q \le M$, the proof follows.
\end{proof}

We now prove Theorem~\ref{thm:KPR}.

\begin{proof}
The proof is by reduction from max cut, and uses for H2S the $\ell_1$ maximization formulation.

Consider a graph $G(V,E)$ with $n$ vertices and $m$ edges that serves as an input instance for max cut. Orient the edges of $G$ arbitrarily. Our reduction uses an integer parameter $M$ (setting $M$ to be $O(n^2m^2)$ will suffice). We reduce the oriented $G$ into an instance of H2S with $k = Mn$ vectors of dimension $d = Mm$ as follows.

The coordinates of vectors are partitioned into $m$ blocks of $M$ coordinates. Each block corresponds to one edge of $G$. Every vertex $v_i$ of $G$ gives rise to $M$ vectors $v_{i,1}, \ldots , v_{i,M}$. In each of these vectors, in every block $B_e$ that corresponds to edge $e$:

\begin{enumerate}

\item If $v_i$ is the head of $e$ then all entries of $B_e$ are~$+1$.

\item If $v_i$ is the tail of $e$ then all entries of $B_e$ are~$-1$.

\item If $v_i$ is not incident with $e$, then the entries of $B_e$ in $v_{i,j}$ (for $1 \le j \le M$) are the $j$th codeword of the Hadamard code $H_M$.

\end{enumerate}

{\bf Yes} instances.
Let $(V_1,V_2)$ be the optimal cut for $G$, and suppose that it cuts $c$ edges (necessarily $c > \frac{m}{2}$). Consider the solution to the H2S instance that partitions the vectors into two clusters $X_1$ and $X_2$ in agreement with the partition $(V_1,V_2)$.

The value of the solution can be lower bounded as follows. There are $Mn$ vectors and $m$ blocks, each of size $M$. Consider a block that corresponds to an edge $e$ that is cut. In each of $X_1$ and $X_2$ there is one endpoint of the edge, and this vertex has a monochromatic block that contributes $M^2$ to the $\ell_1$ norm. This might be partially offset by the other blocks. But the block of each vertex not incident with $e$ can offset at most $M^{3/2}$ of the $\ell_1$ norm, by Proposition~\ref{pro:hadamard}. As there are $c$ edges in the cut, the value of the solution is at least $c(2M^2 - (n-2)M^{3/2})$. (The value is in fact higher because blocks corresponding to edges not in the cut also contribute to the $\ell_1$ norm, but we ignore this further tightening of the parameters.)

{\bf No} instances.
Suppose now that no cut of $G$ cuts $c-1$ edges. Consider an arbitrary partition of the vectors of the H2S instance into two parts $X_1$ and $X_2$. This partition corresponds to a fractional partition of $G$, where the extent $x_i$ to which a vector $v_i$ is in $V_1$ is equal to the fraction of its vectors that are in $X_1$. Similarly, $1 - x_i$ is the extent to which $v_i$ is in $V_2$. For an edge $e = (v_i,v_j)$,  the extent to which it is cut is $y_e = |x_i - x_j|$ (which of course equals $|(1 - x_i) - (1 - x_j)|$).

For an arbitrary edge $e$, consider the contribution of the blocks associated with it to the $\ell_1$ norm. The combination of two monochromatic blocks that are associated with its end points contribute $M^2 y_e$ to the $\ell_1$ norm of each of $X_1$ and $X_2$. The Hadamard blocks associated with vertices that are not end points of $e$ each contributes at most $\sqrt{2}M^{3/2}$ towards the sum of $\ell_1$ norms of $X_1$ and $X_2$. (There is a multiplier of $\sqrt{2}$ rather than just $1$ because a vertex may be split among both sides of the cut. Proposition~\ref{pro:hadamard} allows one to upper bound the effect of this split by $\sqrt{2}$.) Summing up over all edges and all blocks, the value of any solution is at most $2M^2\sum_e y_e + \sqrt{2}(n-2)mM^{3/2}$.

To bound $\sum_e y_e$, observe that local search can always change a fractional cut into an integer cut which is at least as large. Hence $\sum_e y_e \le c-1$.

{\bf Summary.} Subtracting the upper bound for {\em no} instances from the lower bound for {\em yes} instances, it follows that the {\em yes} instance leads to higher value than the {\em no} instance if $2M^2 > (\sqrt{2}m+c)(n-2)M^{3/2}$. Taking $M > 2m^2n^2$ suffices.
\end{proof}

{\bf Remark:} The value of $M$ in the proof of Theorem~\ref{thm:KPR} can be reduced to $O(n^2)$ by using the fact that max-cut is APX-hard. By the results of~\cite{hastad01}, for {\em no} instances we may assume that there is no cut with $0.942 c$ edges.

\subsection*{Acknowledgements}

The question of NP-hardness of the hypergraph 2-segmentation problem was communicated to me by Shai Ben-David and Ruth Urner at the Dagstuhl seminar 14372, {\em Analysis of Algorithms Beyond the Worst Case}, September 2014, where the work reported here was also performed.

Work supported in part by the Israel Science Foundation
(grant No. 621/12) and by the I-CORE Program of the Planning and Budgeting Committee and the Israel Science Foundation (grant No. 4/11).

\end{document}